\documentclass[]{template}

\graphicspath{{figs/}}

\usepackage{amsmath, mathtools, epigraph, framed, mdframed, xcolor}
\usepackage{algorithm}
\usepackage{algpseudocode}
\usepackage{amsthm}
\usepackage{graphicx}
\usepackage{eucal}
\usepackage{mathrsfs}
\usepackage{epigraph}

\newcommand*{\Scale}[2][4]{\scalebox{#1}{\ensuremath{#2}}}%

\theoremstyle{definition}
\newtheorem{definition}{Definition}[section]
\newtheorem{theorem}{Theorem}[section]
\newtheorem{lemma}[theorem]{Lemma}

\begin{document}

\newgeometry{bottom=1.5in}

\begin{center}

  \title{\textsc{The \#ETH is False\\$\# k$-SAT is in Sub-Exponential Time}}
  \maketitle 

  \thispagestyle{empty}
  
  \vspace*{.2in}

  \begin{tabular}{cc}
    Giorgio Camerani\upstairs{\affilone}
   \\[0.25ex]
   {\small Rome, Italy - 2 February 2021}\\
  \end{tabular}
  \date{\today}
  \emails{
    \upstairs{\affilone}giorgio.camerani@gmail.com 
    }
  \vspace*{0.4in}

\textcolor[RGB]{50,50,50}{\rule{290pt}{0.1pt}}
\begin{abstract}
We orchestrate a randomized algorithm for \#$k$-SAT which counts the exact number of satisfying assignments in $2^{o(n)}$ time. The existence of such algorithm signifies that the \#ETH is hereby refuted, and so are $\oplus$ETH, ETH, \#SETH, $\oplus$SETH and SETH.
\end{abstract}
\end{center}

\vspace*{0.15in}
\hspace{18.5pt}
  \Small	


\vspace*{0.5in}

\setlength{\epigraphwidth}{0.57\textwidth}
\epigraph{\footnotesize{\textit{One of the most useful principles of enumeration in discrete probability and combinatorial theory is the celebrated principle of inclusion–exclusion. When skillfully applied, this principle has yielded the solution to many a combinatorial problem.}}}{\footnotesize{\textsc{Gian-Carlo Rota}}}

\section{Introduction}
\label{sec:intro}
\noindent In our previous paper \cite{LSR} we have presented a simple \textit{deterministic} algorithm $\mathscr{A}_0$ for \textit{random} \#$k$-SAT, which counts the \textit{exact} number of satisfying assignments in $2^{\varepsilon n}$ time, with $\lim_{k \to \infty} \varepsilon = 0$ as long as $\Delta = \frac{m}{n} \in 2^{o(k)}$, where $n$ is the number of variables and $m$ is the number of clauses. When $k$ was allowed to grow with $n$ rather than remaining constant, this led us to a sub-exponential time algorithm. The existence of $\mathscr{A}_0$ revealed to us that, at least in the $\Delta \in 2^{o(k)}$ realm, the hardness of random instances decreases as $k$ increases: the longer the clause length, the shorter the running time. The algorithm runs faster and faster as $k$ gets higher and higher. The key insight to obtain such behaviour was the possibility, thanks to the inclusion-exclusion principle, to count satisfying assignments without even searching for them.
\\\\
This paper is devoted to improve $\mathscr{A}_0$ by gradually eliminating its $3$ points of weakness: the $\Delta \in 2^{o(k)}$ restriction, the non-constant $k$ restriction, and the random restriction. In the end, such gradual improvements will culminate into a more general \textit{sub-exponential} time algorithm for \#$k$-SAT, able to deal with any number $m$ of clauses, with any \textit{constant} $k$, and with \textit{worst-case} instances as well.
\\\\
The existence of such algorithm will constitute a single shot confutation of all \#ETH, $\oplus$ETH, ETH, \#SETH, $\oplus$SETH and SETH.
\newpage
\subsection{Contents}The rest of this paper is organized as follows:
\\\\
\begin{tabular}{p{2cm}p{12cm}}
\textbf{Section 2} & Conceives a \textit{deterministic} algorithm $\mathscr{A}_1$ for any \textit{random} \#$k$-SAT instance, which computes the exact counting of satisfying assignments in time $2^{\varepsilon n}$, where\footnotemark{} $\varepsilon \in \Theta(\frac{\log k }{k})$. The clause density $\Delta$ is no longer present in the expression of $\varepsilon$ as it was the case in $\mathscr{A}_0$, thus $\mathscr{A}_1$ works for any $\Delta$, be it critical or dense: the $\Delta \in 2^{o(k)}$ weakness exhibited by $\mathscr{A}_0$ is therefore circumvented here by $\mathscr{A}_1$. As $\lim_{k \to \infty} \varepsilon = 0$, the existence of such counting algorithm is already enough to refute ETH on random $k$-SAT, reason being that ETH is known to imply \cite{ETH} that $\varepsilon$ increases infinitely often as $k \to \infty$, whereas here in reality $\varepsilon$ is monotonically strictly decreasing. \\
 & \\
\textbf{Section 3} & Uses $\mathscr{A}_1$ to devise a more general \textit{randomized} algorithm $\mathscr{A}_2$, working for \textit{any} \#$k$-SAT instance with \textit{constant} $k$, which counts the exact number of satisfying assignments in time $2^{O\left(\frac{\log \log \log n}{\log \log n}n\right)}$. Such final \textit{sub-exponential} time algorithm $\mathscr{A}_2$ eliminates both the dependency on $k$ in the running time, and the random restriction on the input formula. The randomness here is only used to turn the input formula into a formula which \textit{looks} random to $\mathscr{A}_1$. \\
\end{tabular}
\footnotetext{All the logarithms in this paper are base $2$. Moreover, we voluntarily omit polynomial factors, in order to not encumber the aestethics: each time we write $2^t$ we actually mean $O^\star(2^t)$, where the $O^\star$ notation suppresses potentially existing factors of magnitude at most polynomial in the instance size. Finally, an additional pedantic statement: it is self-evident that the variable in such $\Theta$ notation is $k$, certainly not $n$.}
\\
\subsection{A quick note on randomization}The literature is not unanimous on whether ETH and its close relatives allow randomized algorithms: in some works they do, in some they do not, in some others no explicit statement is made. In \cite{CountETH}, where \#ETH has been introduced, it is explicitly restricted to deterministic algorithms only. However in \cite{ETH}, where ETH was originally formulated, no such restriction is made\footnote{$s_k = \inf\{\delta: \exists\ 2^{\delta n}\text{ algorithm for solving }k\text{-SAT}\}$.} and in \cite{SETH}, where SETH was introduced, randomized algorithms are explicitly allowed\footnote{$s_k = \inf\{\delta:\exists\text{ randomized algorithm for }k\text{-SAT with time complexity poly}(m)2^{\delta n}\}$.} for both ETH and SETH. We therefore feel justified in adopting the assumption that all these hypotheses permit the usage of randomized algorithms. We see no reasonable motivation for forbidding them. 

\subsection{A quick note on $k$}The literature is also not unanimous on what $k$ means: whether each clause has \textit{exactly} $k$ literals, or \textit{at most} $k$ literals. Papers on random $k$-SAT adopt the former definition, while papers on ETH use the latter. Consistently, in Section 2 we assume $=k$, while in Section 3 we will assume $\leq k$.

\section{Solving Random \#$k$-SAT in $2^{\Theta\left(\frac{\log k}{k}\right)n}$ Deterministic Time}\label{sec:A1}
\noindent Let $\Phi = \Phi( n, m, k )$ be a $k$-CNF formula on $n$ boolean variables with $m$ clauses, each one having length \textit{exactly} $k$ and being chosen uniformly at random among the $2^k {n \choose k}$ possible candidates. See how such definition forbids the usage of the same variable more than once in the same clause. In this section we are going to excogitate a deterministic algorithm $\mathscr{A}_1$ for counting the exact number of satisfying assignments of any such $\Phi$ in time $2^{\varepsilon n}$ where $\varepsilon = \Theta(\frac{\log k}{k})$. The probability that the returned counting is wrong is $\approx \frac{1}{n^{\sigma \log e}}$, where the integer constant $\sigma \geq 1$ is just a tuning parameter of the algorithm which as such does not depend on the input instance (neither on $n$ nor on $k$), and which let us control such probability. The reader might be wondering how is it possible that the algorithm is deterministic, yet it has an error probability greater than zero. This is the vanishingly small price we have to pay in order to eliminate the $\Delta \in 2^{o(k)}$ restriction that $\mathscr{A}_0$ had: as we will see, the only difference between $\mathscr{A}_1$ and $\mathscr{A}_0$ is just a simple observation which allows us to ignore a massive portion of the search space\footnote{Recall that, as shown in \cite{LSR}, the search space here is not the space of satisfying assignments, but the space of monotone sub-formulae.}: such ignored portion is so massive that it let us eliminate the dependency of the exponent from $\Delta$, and the aforementioned error probability is the probability that ignoring it will jeopardise the correct counting of satisfying assignments.
\subsection{Notations and definitions} Let $\Phi = \{c_1, \cdots, c_m\}$, and let $V = \{v_1, \cdots, v_n\}$ be the set of variables of $\Phi$. Each clause $c_i = \{\ell_{i,1}, \cdots, \ell_{i,k}\}$ is a set of literals, where each literal is either a variable $v \in V$ or its negation. Let $\mathcal{A} = \{v_1, \lnot v_1\} \times \cdots \times \{v_n, \lnot v_n\}$ denote the set of all the $2^n$ possible boolean assignments to the $n$ variables in $V$. Let $\mathcal{S} = \{ b \in \mathcal{A} : \forall c \in \Phi \ c \cap b \neq \varnothing \}$ be the set of satisfying assignments of $\Phi$. Let $\mathcal{U} = \mathcal{A} \setminus \mathcal{S} = \{ b \in \mathcal{A} : \exists c \in \Phi \ c \cap b = \varnothing \}$ be the set of unsatisfying assignments of $\Phi$.

\theoremstyle{definition}
\begin{definition}[\textit{Sub-formula} of $\Phi$]
A \textit{sub-formula} $\Psi$ of $\Phi$ is any formula $\Psi \subseteq \Phi$.
\end{definition}

\theoremstyle{definition}
\begin{definition}[\textit{Monotone} formula]
A formula is \textit{monotone} if and only if each of its variables always appears with the same sign: either always positive or always negated\footnote{See how for a formula to be monotone it is not required that all the variables carry the same sign. Different variables can have different signs. The only restriction is that every same variable always carries the same sign.}. 
\end{definition}

\theoremstyle{definition}
\begin{definition}[\textit{Compatible} clause]
Given a monotone formula $\Psi$ and a clause $c \notin \Psi$, we say that $c$ is \textit{compatible} with $\Psi$ if and only if $\Psi \cup \{c\}$ is still monotone\footnote{We also say that $c$ is \textit{incompatible} with $\Psi$ if and only if $c$ is not compatible with $\Psi$.}. 
\end{definition}

\theoremstyle{definition}
\begin{definition}[\textit{Maximal} monotone sub-formula]
A monotone sub-formula $\Psi$ of $\Phi$ is \textit{maximal} if and only if $\forall c \in \Phi \setminus \Psi$ it is the case that $c$ is incompatible with $\Psi$.
\end{definition}

\subsection{How $\mathscr{A}_0$ worked} Let $O_\nu$ (respectively $E_\nu$) be the number of monotone sub-formulae of $\Phi$ having $\nu$ variables and an odd (respectively even) number of clauses. In \cite{LSR} we have proven the following: 
\begin{theorem}
\begin{equation}\label{eq:property}
|\mathcal{S}| = 2^n - \sum_{\nu = 1}^{n} ( O_\nu - E_\nu ) \cdot 2^{n-\nu}    
\end{equation}
\end{theorem}
\noindent The above identity\footnote{See how it can be further shrinked to $|\mathcal{S}| = \sum_{\nu = 0}^{n} ( E_\nu - O_\nu ) \cdot 2^{n-\nu}$.} shows how the number $|\mathcal{U}|$ of unsatisfying assignment of any\footnote{Recall that Theorem \ref{eq:property} holds for as \textit{generic} as possible CNF expressions.} CNF formula can be expressed as a function of the space of its monotone sub-formulae. This meant we could count satisfying assignments by merely enumerating all the monotone sub-formulae of the input instance $\Phi$, without even trying to search for a single one satisfying assignment: we can \textit{count without search}. We have seen how, as $k \to \infty$, for \textit{random} \#$k$-SAT instances the cardinality $|\Psi|$ that any \textit{maximal} monotone sub-formula $\Psi \subset \Phi$ can possibly have grows at most as $\frac{n}{k}$, as long as $\Delta \in 2^{o(k)}$. This led us to devise $\mathscr{A}_0$, which perlustrated the whole space of monotone sub-formulae of $\Phi$ by simply brute-forcedly enumerating all the subsets of at most $\frac{n}{k}$ clauses picked from the $m = \Delta n$ available clauses. The number of such subsets is $2^{\Theta\left(\frac{\log( \Delta k )}{k}\right)n}$, hence the running time of $\mathscr{A}_0$. The Achilles' heel of such approach is that when $\Delta \notin 2^{o(k)}$, e.g. as it is the case for the hardest random instances existing at the critical threshold \cite{coja-oghlan} around $2^k$, the running time of $\mathscr{A}_0$ spirals out of control as $k$ grows, possibly behaving even worse than naïve exhaustive search\footnote{The main reason of such collapse in performances is that, when $\Delta \notin 2^{o(k)}$, it is no longer necessarily true that every monotone sub-formula $\Psi$ has at most $\approx\frac{n}{k}$ clauses. Even when it is still the case, like with $\Delta = 2^{\alpha k}$ for $\alpha < \frac{1}{2}$, the final running time is $2^{(\alpha + \frac{\log k}{k})n}$, thus it is no longer true that $\lim_{k\to\infty} \varepsilon = 0$.}. To overcome such fatal vulnerability, some new insight was needed.

\subsection{A new insight: how to prune the search space} In order to illustrate the crucial, yet very simple, observation that lets us able to ignore a remarkable portion of the space of monotone sub-formulae we have to explore, let us imagine we are visiting such whole space by starting with $\Psi = \varnothing$ and by scanning the $m$ clauses sequentially. For each $c_i \in \Phi$ compatible with $\Psi$, we branch: either we add it to $\Psi$, or we do not. Each time we add such a $c_i$ to $\Psi$, the number of its variables increases by an amount between $0$ and $k$. 

\theoremstyle{definition}
\begin{definition}[\textit{Fruitless} clause]
Given a monotone sub-formula $\Psi \subset \Phi$ and a clause $c_j \notin \Psi$ compatible with $\Psi$, we say that $c_j$ is \textit{fruitless} for $\Psi$ if and only if $\Psi$ and $\Psi \cup \{c_j\}$ have the same number of variables. 
\end{definition}

\noindent In other words, a fruitless clause is just a clause that, should it be added to $\Psi$, would not bring any new variable to it: all the $k$ literals it has are already mentioned in $\Psi$. Now, our observation is as simple as this:
\\
\begin{mdframed}[backgroundcolor=gray!7] 
\vspace*{\fill} 
\begin{quote} 
\centering 
\textit{Adding a fruitless clause to $\Psi$ is a completely useless operation, which does not affect at all the counting of the exact number of satisfying assignments of $\Phi$.} 
\end{quote}
\vspace*{\fill}
\end{mdframed}
\vspace{10pt}
\noindent To intuitively see why such observation is correct, think about this: as the number $\nu$ of variables of $\Psi$ is equal to the number of variables of $\Psi \cup \{c_j$\}, this means that, in \ref{eq:property}, $\Psi$ will be counted among the $O_\nu$ and $\Psi \cup \{c_j\}$ among the $E_\nu$, or vice-versa: they \textit{cancel out} each other, bringing a null and void contribution to $|\mathcal{U}|$. See how far can this go: while gradually assembling a monotone sub-formula $\Psi$, as soon as we detect the existence of a fruitless clause $c_j$ ahead, it does not \textit{just} mean that $\Psi$ and $\Psi \cup \{c_j\}$ \textit{only} are useless for our counting purpose. It also means that the entire $\Psi$ built so far is totally useless, and there is no need to go on any further: just because $c_j$ exists, it makes no sense to visit \textit{any} of the remaining clauses still to be considered. The mere existence of such a $c_j$ \textit{invalidates} $\Psi$ as a whole. To epitomize such intuition:
\\
\begin{mdframed}[backgroundcolor=gray!7] 
\vspace*{\fill} 
\begin{quote} 
\centering 
\textit{Just because $\exists c_j$ fruitless for $\Psi$, every $\Psi' \supseteq \Psi$ brings no contribution to $|\mathcal{U}|$.} 
\end{quote}
\vspace*{\fill}
\end{mdframed}
\vspace{10pt}
\noindent Every larger $\Psi'$ having $\Psi$ as a subset would be subject to such very same annihilation: for $c_j$ could be either present or absent in $\Psi'$ without affecting its number $\nu'$ of variables, thereby causing a worthless $+1 -1$ contribution to the quantity $O_{\nu'} - E_{\nu'}$. This means that, as soon as we determine that such a fruitless $c_j$ exists \textit{somewhere down there} in the remaining sequence of clauses yet to be considered, we can legitimately stop here and overthrow $\Psi$ wholly, thereby pruning the search space by ignoring \textit{all} the $\Psi' \supseteq \Psi$: their entire recursion sub-tree, rooted at $\Psi$, gets discarded without even being perlustrated. We can then rollback to the last monotone sub-formula we had before $\Psi$, continuing the recursion from there onwards. 

\begin{figure}[H]
    \centering
    \includegraphics[scale=0.1355]{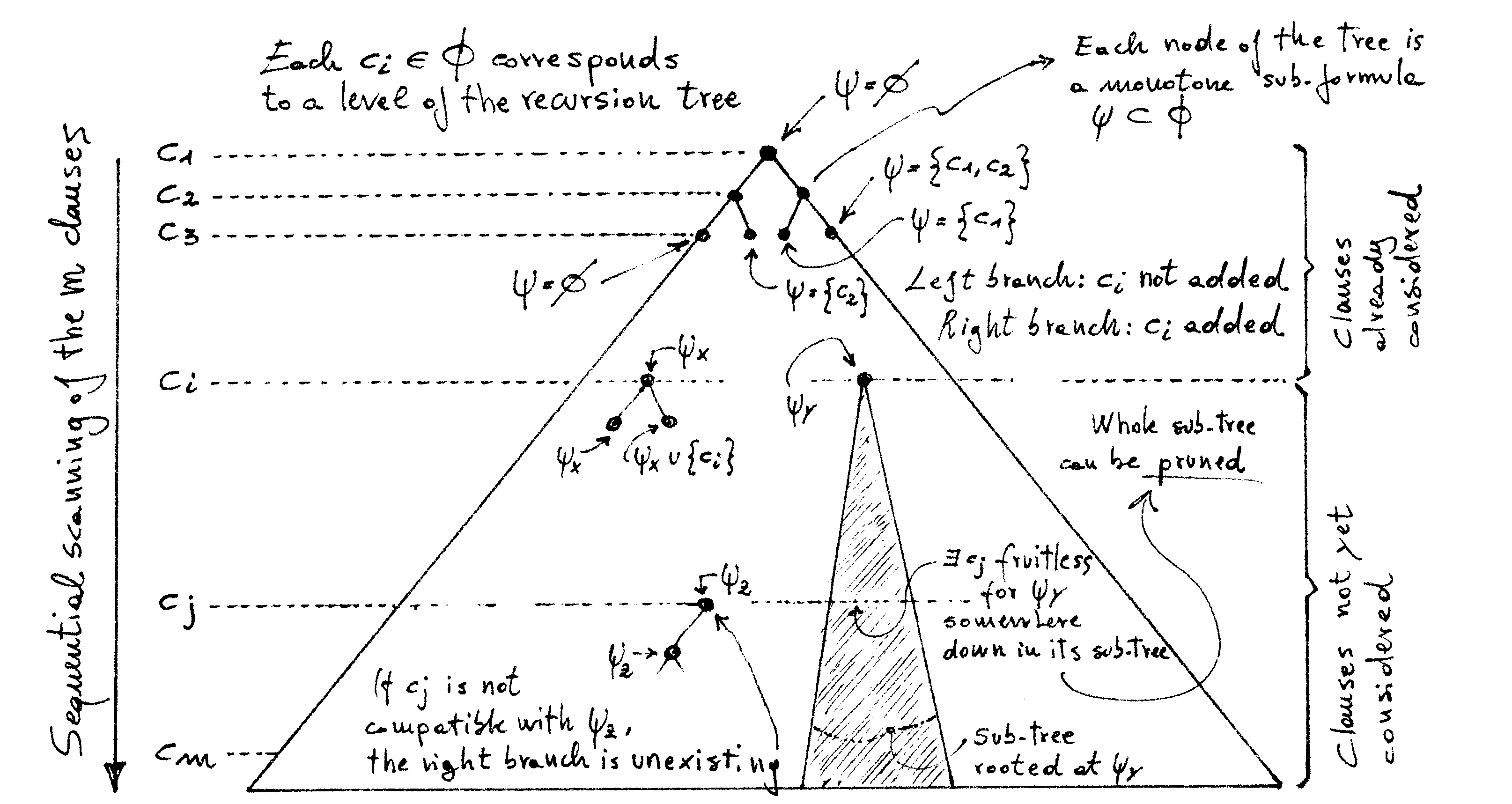} \\
    \caption{Visiting and pruning the search space of monotone sub-formulae}
    \label{fig:pruning}
\end{figure}

\noindent The above figure offers a visual imprinting of the narration held so far. See how every node of the recursion tree corresponds to a certain monotone sub-formula $\Psi \subset \Phi$. Which one? The $\Psi$ we have built down to that point, thanks to the choices we have made in the ancestors nodes, while deciding whether to add or not each of the compatible clauses so far considered (see how when a clause is incompatible, there is no choice to be made as we only have the left branch). We conclude here by stressing once again this crucial detail: as Figure \ref{fig:pruning} suggests, the fruitless $c_j$ which let us completely disregard the sub-tree rooted at $\Psi$ needs not to be at the same level where $\Psi$ itself is situated: it can be \textit{anywhere down} such sub-tree, any arbitrary number of levels below. With such consideration clear in mind, we step into describing $\mathscr{A}_1$.

\subsection{How $\mathscr{A}_1$ works} Let ourselves be wandering somewhere in the recursion tree of the search space of monotone sub-formulae. Let us be standing on a certain node $\Psi$ of such tree. Let $\nu_\Psi$ be the number of variables that $\Psi$ has.

\theoremstyle{definition}
\begin{definition}[\textit{Saturation}]
$s_\Psi = \frac{\nu_\Psi}{n}$. 
\end{definition}

\noindent The saturation $s_\Psi$ of a monotone sub-formula $\Psi$ is just a number comprised between $0$ and $1$ which represents the amount of variables that $\Psi$ has, compared to the overall number of variables mentioned in the input formula $\Phi$. As we keep adding clauses, the saturation clearly grows. Let $L_\Psi \in [1, \cdots, m]$ be the level of the tree where $\Psi$ is situated. 

\theoremstyle{definition}
\begin{definition}[\textit{Pruning probability}]
The probability $\mathcal{P}_\Psi^p$ that $\exists c_j$ fruitless for $\Psi$ with $j \geq L_\Psi$. 
\end{definition}

\noindent Let $T_\Psi$ denote the sub-tree rooted at $\Psi$. The pruning probability $\mathcal{P}_\Psi^p$ is thus the probability\footnote{The $p$ in the superscript of $\mathcal{P}_\Psi^p$ is a mnemonic for \textit{pruning}.} that the whole $T_\Psi$ can be disregarded without even being scrutinized by $\mathscr{A}_1$, due to the existence of \textit{at least one} fruitless clause for $\Psi$ at \textit{any} level of $T_\Psi$. We are now ready to formulate the following, naturally arising question:
\\
\begin{mdframed}[backgroundcolor=gray!7] 
\vspace*{\fill} 
\begin{quote} 
\centering 
\textit{As the saturation $s_\Psi$ grows, how does the pruning probability $\mathcal{P}_\Psi^p$ evolve?} 
\end{quote}
\vspace*{\fill}
\end{mdframed}
\vspace{10pt}
 Our aim would be to express $\mathcal{P}_\Psi^p$ as a function of $s_\Psi$. To do so, we need to introduce the following first:

\theoremstyle{definition}
\begin{definition}[\textit{Fruitless probability}]
The probability $\mathcal{P}_\Psi^f$ that a randomly picked clause is fruitless for $\Psi$. 
\end{definition}

\noindent The following lemma relates the fruitless probability\footnote{The $f$ in the superscript of $\mathcal{P}_\Psi^f$ is a mnemonic for \textit{fruitless}.} to the saturation, and will be used to compute $\mathcal{P}_\Psi^p$: 
\begin{lemma}
\begin{equation}\label{eq:p_Psi^f}
\mathcal{P}_\Psi^f \approx \frac{1}{2^{k(1-\log s_\Psi)}}    
\end{equation}
\end{lemma}
\begin{proof}
Each compatible clause that we might add to $\Psi$ has $0$ to $k$ literals in common with $\Psi$ itself. In order for such a randomly generated clause to be fruitless, it has to have all of its literals in common, thus all the $k$ of them shall be picked among the $s_\Psi n$ variables of $\Psi$, with the same signs (otherwise it would be incompatible). $\mathcal{P}_\Psi^f$ can be thus expressed as the ratio between favourable outcomes (fruitless clauses) and all outcomes (all available clauses):
\[ \mathcal{P}_\Psi^f = \frac{ { s_\Psi n \choose k } }{2^k {n \choose k} } \]

\noindent Using Stirling's approximation\footnote{$\log {a \choose b} \approx b \log \frac{a}{b} + ( a - b )\log \frac{a}{a-b} = b \log \frac{a}{b} - ( a - b )\log(1 - \frac{b}{a})$. As $\log( 1 - x ) \approx -x$ for small $x$, the expression finally simplifies to $\log {a \choose b} \approx b \log \frac{a}{b} + b - \frac{b^2}{a}$. This holds whichever the base of the logarithm is, because all the terms such as $-x\log e$ cancel out each other.}, we can write:
\[ \log \mathcal{P}_\Psi^f \approx k \log \frac{s_\Psi n}{k} + k -\frac{k^2}{s_\Psi n} - k - k \log \frac{n}{k} - k + \frac{k^2}{n} = k \log s_\Psi - k + -\frac{k^2}{s_\Psi n} + \frac{k^2}{n} \]

\noindent The $2$ rightmost terms monotonically descend towards $0$ as $n\to\infty$ and as we keep adding clauses to $\Psi$, they can therefore be ignored asymptotically, leading us to the following expression which closes the proof:
\[
\log \mathcal{P}_\Psi^f \approx k ( \log s_\Psi - 1 ) = -k ( 1 - \log s_\Psi )    
\]
\noindent We conclude with an observation: the smallest term $\frac{1}{2}\log(2\pi x)$ coming from the Stirling approximation, which was obviously ignored in the above reasoning, would have given an overall contribution of $\frac{1}{2}\log( \frac{ s_\Psi(n-k) }{s_\Psi n - k } )$ which evidently collapses to $0$ as $n \to\infty$.
\end{proof}
\noindent We are interested in how $\mathcal{P}_\Psi^f$ behaves as we keep halving $s_\Psi$. For each $h \geq 1$, let us set:
\begin{equation}\label{eq:s_Psi}
s_\Psi = 2^{-h + 1}    
\end{equation} 

\noindent That is to say, we start with a saturation $s_\Psi = 1$ (corresponding to $h = 1$) and we keep increasing $h$: each time $h$ is increased by $1$, the saturation gets halved. Plugging \ref{eq:s_Psi} into \ref{eq:p_Psi^f} gives:
\begin{equation}
\mathcal{P}_\Psi^f \approx \frac{1}{2^{kh}} 
\end{equation} 
  
\noindent Now we set $h = \frac{\log n - \log \log n}{k}$ and see what happens to both the saturation:
\begin{equation}\label{eq:s_Psi_h}
s_\Psi = \frac{2 \log^\frac{1}{k}n}{n^\frac{1}{k}} 
\end{equation} 

\noindent and the fruitless probability:
\begin{equation}\label{eq:p_fruitless}
\mathcal{P}_\Psi^f \approx \frac{\log n}{n} 
\end{equation} 

\noindent Let us call \ref{eq:s_Psi_h} the \textit{critical} saturation. Gluing it together, here is what it all \textit{roughly} means, on average:
\\
\begin{mdframed}[backgroundcolor=gray!7] 
\vspace*{\fill} 
\begin{quote} 
\centering 
\textit{Once we have reached the critical saturation, we should expect\\ to find $\log n$ fruitless clauses every further $n$ clauses we scan.}
\end{quote}
\vspace*{\fill}
\end{mdframed}
\vspace{10pt}

\noindent We are now ready to come back to $\mathcal{P}_\Psi^p$, and to formulate an expression for it telling us how it behaves in correspondence of the critical saturation as $n\to\infty$.
\begin{lemma}
If $\Psi$ has saturation at least critical, the following holds:
\begin{equation}\label{eq:p_Psi^p}
\lim_{n\to\infty} \mathcal{P}_\Psi^p = 1    
\end{equation}
\end{lemma}
\begin{proof}
We focus on the \textit{last}\footnote{As $\Phi$ is a \textit{set} of clauses, there is no notion of \textit{last}. We therefore mean \textit{last} with respect to a certain ordering. Which ordering? Say, the ordering the $m$ clauses have been randomly picked in the first place, from the set of $2^k{n \choose k}$ candidates.} $\sigma n$ clauses of $\Phi$ and consider the Bernoulli process $X_1, \cdots, X_{\sigma n}$ where $X_i$ is the random variable defined as follows:
\begin{equation*}
    X_i = \begin{cases}
               1               & \text{if the }i\text{-th clause is fruitless for } \Psi\\
               0               & \text{otherwise}
           \end{cases}
\end{equation*}
\noindent Clearly we mean the $i$-th clause of $\Phi$ \textit{among its last} $\sigma n$ clauses: we are standing somewhere on a certain $\Psi$ of the recursion tree, from our node having critical saturation we look down toward the last $\sigma n$ layers of the tree and conduct our Bernoulli process on them. By \ref{eq:p_fruitless}, $X_i = 0$ with probability $\approx 1-\frac{\log n}{n}$. The random variable $X = \sum_{i = 1}^{\sigma n} X_i$ tracks the number of fruitless clauses for $\Psi$ among the last $\sigma n$ clauses of $\Phi$, and the pruning probability $\mathcal{P}_\Psi^p$ is clearly at least equal to the probability that $X > 0$ (in general $\mathcal{P}_\Psi^p$ is higher than that, because a fruitless clause might very well exist also before the last $\sigma n$ clauses). Considering that the probability that $X = 0$ is:
\[
\mathbf{P}(X = 0) \approx \left (1-\frac{\log n}{n} \right)^{\sigma n} 
\]
and that $\lim_{n \to \infty}\mathbf{P}(X=0) = \frac{1}{n^{\sigma \log e}}$ this obvious conclusion follows asymptotically:
\[
\mathcal{P}_\Psi^p \geq 1 - \frac{1}{n^{\sigma \log e}}
\]
which trivially means $\lim_{n \to \infty} \mathcal{P}_\Psi^p = 1$, thereby concluding the proof.
\end{proof}
\noindent That is to say: at or above the critical saturation, every sub-tree $T_\Psi$ is asymptotically almost surely prunable. The probability that there are no fruitless clause among the last $\sigma n$ clauses drops to $0$ as a power of $n$, and we can control such power by tuning $\sigma$ at our will. Once $\Psi$ reaches the critical saturation, the probability that the sub-tree $T_\Psi$ rooted at $\Psi$ is totally worthless to be explored, and can therefore be ignored tout-court without being visited and without affecting at all the correctness of the final counting, quickly approaches $1$ as $n \to \infty$. Epitomizing it:
\\
\begin{mdframed}[backgroundcolor=gray!7] 
\vspace*{\fill} 
\begin{quote} 
\centering 
\textit{It makes literally no sense to use the first $m - \sigma n$ clauses to build\\monotone sub-formulae having a saturation higher than the critical.\\We can enumerate them only up to the critical saturation, not more.}
\end{quote}
\vspace*{\fill}
\end{mdframed}
\vspace{10pt}

\noindent Doing so we are going to avoid the exploration of a \textit{massive} portion of the search space, because each of the $\Psi$ we are going to consider will be built as follows: by picking \textit{few} clauses from the $m-\sigma n$ side, and combining them with clauses picked from the $\sigma n$ side (where $\mathscr{A}_0$ works nicely, because $\sigma$ is independent of $k$). We are now going to formalize such intuitive statement: firstly by ending this sub-section presenting the pseudo-code of $\mathscr{A}_1$, and secondly by proving its running time immediately after.
\renewcommand{\thealgorithm}{}
\begin{algorithm}
\caption{$\mathscr{A}_1$ Computes the exact number of satisfying assignments of random $\Phi$}
\label{alg:A1}
\begin{algorithmic}[1]
\Procedure{CountRandom}{$\Phi,\sigma$}
    \State Let $\Phi_\uparrow$ be the sub-formula of $\Phi$ obtained by selecting its first $m-\sigma n$ clauses 
    \State Let $\Phi_\downarrow$ be the sub-formula of $\Phi$ obtained by selecting its last $\sigma n$ clauses
    \State Initialize $\langle \nu, O_{\nu}, E_{\nu} \rangle \gets \langle \nu, 0, 0 \rangle$, $\forall \nu \in [\ k,\ n\ ]$
        \For {each monotone sub-formula $\Psi_\uparrow$ of $\Phi_\uparrow$ having saturation $s_{\Psi_\uparrow}$ less than critical}
            \For {each monotone sub-formula $\Psi_\downarrow$ of $\Phi_\downarrow$}
                \State Let $\Psi = \Psi_\uparrow \cup \Psi_\downarrow$ be monotone and have $\nu$ variables
                \If {$|\Psi|$ is odd}
                    \State $\langle \nu, O_{\nu}, E_{\nu} \rangle \gets \langle \nu, O_{\nu} + 1, E_{\nu} \rangle$
                \Else 
                    \State $\langle \nu, O_{\nu}, E_{\nu} \rangle \gets \langle \nu, O_{\nu}, E_{\nu} + 1 \rangle$
                \EndIf
            \EndFor
        \EndFor
	\State $count \gets 0$
	\For{$\nu \in [\ k,\ n\ ]$}
	    \State $count \gets count + ( O_{\nu} - E_{\nu} ) \cdot 2^{n-\nu}$
	\EndFor
	\State Return $2^n-count$
\EndProcedure
\end{algorithmic}
\end{algorithm}

\noindent Some quick observations, all pretty obvious. Firstly, it's clear that, at line $7$, if $\Psi$ does not happen to be monotone we just skip it and go on with the inner iteration, to the next $\Psi_\downarrow$. Secondly, it's also clear that the above algorithm $\mathscr{A}_1$ is not optimal: taking into account the existence of fruitless clauses only among the last $\sigma n$, whereas a fruitless clause might very well exist at any index after the last clause of $\Psi_\uparrow$, is a rough simplification. However, such simplification renders $\mathscr{A}_1$ more amenable to a straightforward running time analysis: as we are about to see, the resulting asymptotic running time will not depend on $\Delta$ in the end, so sufficient for our purpose. Thirdly, see how the above algorithm is deterministic, as there is no usage of random bits in it, yet it is not such an algorithm stricto sensu either: rather, it returns the correct answer with probability $1$, asymptotically almost surely. The probability that the returned answer might be wrong, which drops to $0$ as $n \to \infty$, is due to the remote, vanishing possibility that there might exist some $\Psi_\uparrow$ having saturation equal or higher than critical, yet $T_{\Psi_\uparrow}$ being not prunable due to nonexistence of fruitless clauses for $\Psi_\uparrow$ (meaning that it might furnish a non-null contribution to the final counting).

\subsection{Proof of $\mathscr{A}_1$ running time}
We are now ready to prove the following:

\begin{theorem}\label{the:A1}
$\mathscr{A}_1$ runs in $2^{\varepsilon n}$ time, with $\lim_{k \to \infty} \varepsilon = 0$. 
\end{theorem}
\begin{proof}
We are going to show that $\varepsilon \in \Theta(\frac{\log k}{k})$. First of all, we call $\Psi_\uparrow \subset \Phi$ critical if and only if its saturation $s_{\Psi_\uparrow}$ is equal to the critical saturation. We shall then ask the following natural question: how many clause does a critical $\Psi_\uparrow$ have? We need an upper bound on the number of clauses that such a critical $\Psi_\uparrow$ might possibly exhibit. By repeating a very basic reasoning already held in \cite{LSR}\footnote{See Theorem 4.1 over there.}, it is easy to see that, due to the \textit{random} nature of $\Phi$, each time we add a new compatible clause of $k$ literals to the $\Psi_\uparrow$ we are assembling, \textit{at least} $\frac{k}{2}$ of such literals will be \text{new} variables for $\Psi_\uparrow$: as we keep adding clauses, we imagine to stack all the growing number of variables of $\Psi_\uparrow$ on the "left" half of $n$, then doing so each randomly generated clause will roughly have half of its literals falling on the "left" side (which occupancy is growing), and half of them on the "right" side (which keeps on remaining empty). This is true as long as $\Psi_\uparrow$ has at most $\frac{n}{2}$ variables. As we have to keep on adding clauses only up to the critical saturation $s_{\Psi_\uparrow} = \frac{2 \log^\frac{1}{k}n}{n^\frac{1}{k}}$, which corresponds to an amount of variables equal to $s_{\Psi_\uparrow} n = 2 n^\frac{k-1}{k} \log^\frac{1}{k}n$, and by the truism that such amount is clearly less than $\frac{n}{2}$ for large enough $n$, we can legitimately assume that in order to build our critical $\Psi_\uparrow$ we need to invest $1$ compatible clause for every $\frac{k}{2}$ variables we want in it, which leads to the following upper bound on the number of clauses of any critical $\Psi_\uparrow$:
\begin{equation}\label{eq:cardMax}
|\Psi_\uparrow| \leq \frac{4 n^\frac{k-1}{k} \log^\frac{1}{k}n}{k}
\end{equation}

\noindent By invoking once again Stirling's approximation as we did in the proof of Lemma \ref{eq:p_Psi^f}, we can determine how many critical $\Psi_\uparrow$ can be assembled by picking clauses among the first $m - \sigma n$ clauses of $\Phi$:

\begin{equation}\label{eq:outer}
\log { (\Delta - \sigma) n \choose \frac{4 n^\frac{k-1}{k} \log^\frac{1}{k}n}{k} } \approx \underbrace{\frac{4 n^\frac{k-1}{k} \log^\frac{1}{k}n}{k} \log \frac{k(\Delta - \sigma)n^\frac{1}{k}}{4 \log^\frac{1}{k} n}}_\text{Leading term}  + \frac{4 n^\frac{k-1}{k} \log^\frac{1}{k}n}{k} - \frac{16 n^\frac{k-2}{k} \log^\frac{2}{k}n}{k^2 (\Delta - \sigma)}
\end{equation}

\noindent See how the leading term belongs to $o(n)$ for \textit{any} clause density $\Delta$. Let us sculpture it more evidently:
\\
\begin{mdframed}[backgroundcolor=gray!7] 
\vspace*{\fill} 
\begin{quote} 
\centering 
\textit{The outer iteration of $\mathscr{A}_1$ cycles }$2^{o(n)}$\textit{ many times, whatever the $\Delta$ is.}
\end{quote}
\vspace*{\fill}
\end{mdframed}
\vspace{10pt}

\noindent It must be observed that the expression in \ref{eq:outer} is a gross overestimation, because we are considering all the sub-formulae of $\Phi$, even the non-monotone ones: reality is therefore much better than that. We now focus on the inner iteration of $\mathscr{A}_1$, that is to say on the last $\sigma n$ clauses. Since $\sigma$ is a constant which does not depend on $k$, we can apply the result we have already proven in \cite{LSR} for $\mathscr{A}_0$, where we have shown that as long as $\frac{\sigma n}{n} \in 2^{o(k)}$ it is the case that every \textit{maximal} monotone sub-formula $\Psi$ has at most $\approx \frac{n}{k}$ clauses as $k\to\infty$. This fact allowed us to conclude that $\mathscr{A}_0$ had to perlustrate not more than $m \choose \frac{n}{k}$ such maximal $\Psi$s. By repeating that very same simple reasoning here on the last $\sigma n$ clauses, we can state the following:
\\
\begin{mdframed}[backgroundcolor=gray!7] 
\vspace*{\fill} 
\begin{quote} 
\centering 
\textit{The inner iteration of $\mathscr{A}_1$ cycles at most }$2^{\Theta(\frac{\log k}{k})n}$\textit{ many times.}
\end{quote}
\vspace*{\fill}
\end{mdframed}
\vspace{10pt}

\noindent The above follows from applying Stirling's approximation to $\sigma n \choose \frac{n}{k}$ and observing how the resulting exponent asymptotically behaves as $2^{\varepsilon n}$ with $\varepsilon = \frac{\log(\sigma k)}{k} + \frac{1}{k} - \frac{1}{\sigma k^2} < \frac{\log k}{k} + \frac{\log \sigma}{k} + \frac{1}{k} \in \Theta(\frac{\log k}{k})$. Gluing it all together: we have $2$ nested iterations, the outer one scanning over $2^{o(n)}$ elements (the $\Psi_\uparrow$), and the inner one scanning over $2^{\varepsilon n}$ elements (the $\Psi_\downarrow$), which obviously means the overall number of considered merged objects (the $\Psi = \Psi_\uparrow \cup \Psi_\downarrow$) is given by their product, which translates into the sum of the two (outer and inner) exponents. The amount of times we are going to execute lines from $7$ to $12$ can therefore be written as follows (considering only critical $\Psi_\uparrow$s and maximal $\Psi_\downarrow$s):
\[\Scale[1.25]
{
\Scale[2.10]{2}^{ \overbrace{ \underbrace{\frac{4 \log^\frac{1}{k}n}{k} n^\frac{k-1}{k} \log \frac{k(\Delta -\sigma)n^\frac{1}{k}}{4 \log^\frac{1}{k} n}}_{\text{Negligible term, outer iteration}} }^{o(n)\text{ for any }\Delta\text{ and }k}\ + \ \  \Scale[1.35]{\underbrace{\overbrace{\Theta\left(\frac{\log k}{k}\right)}^{\lim_{k \to \infty}\varepsilon=0}n}_{\substack{\text{Leading term,} \\ \text{inner iteration}}\\}}}
}
\]

\noindent By the very definition of $\Theta$ notation\footnote{That is to say, we pretend that the cheaper outer iteration has the same cost as the inner one.}, we can assert that the total number of steps performed by $\mathscr{A}_1$ is asymptotically upper bounded by (up to a polynomial factor\footnote{Such polynomial factor depends on $3$ sources: the fact that we ignored the $\frac{1}{2}\log(2\pi x)$ term in Stirling's approximation, the fact that line $7$ requires to visit all the $k$ literals, and the fact that we focused on maximal $\Psi$s only whereas we need to enumerate all $\Psi$s \textit{up to} maximal size (very roughly, this can be brutally adjusted thanks to a $\frac{n}{k}$ factor squared, that is to say by pretending that there are, for each size up to the maximal size, as many clauses as there are for the maximal, squaring due to the two nested cycles).}):
\[
\Scale[2.1]{2^{o(n)\Scale[0.5]{\ +\ }\Theta\left(\frac{\log k}{k}\right)n} \Scale[0.75]{\ \in}\ 2^{\Theta\left(\frac{\log k}{k}\right)n}}
\]

\noindent thereby concluding the proof. We have a deterministic algorithm, which probability $\frac{1}{n^{\sigma\log e}}$ of returning a wrong answer collapses to $0$ as $n$ grows, running \textit{faster and faster} on random $k$-SAT instances as $k \to\infty$, for \textit{any} clause density, whatever $\Delta \in 2^{\Theta(k)}$ or $\Delta \in \Theta(n^{k-1})$.
\end{proof}


\section{Solving \#$k$-SAT in $2^{O\left(\frac{\log \log \log n}{\log \log n}n\right)}$ Time with the help of Randomness}\label{sec:A2}
\noindent In this section, we are now going to orchestrate a plot to turn our non-randomized algorithm working on random instances into a randomized algorithm working on non-random instances:
\\
\begin{mdframed}[backgroundcolor=gray!7] 
\vspace*{\fill} 
\begin{quote} 
\centering 
\textit{We derandomize the instance by randomizing the algorithm.}
\end{quote}
\vspace*{\fill}
\end{mdframed}
\vspace{10pt}

\noindent In order to do that, we will devise a randomized reduction which, given a generic formula $\Phi$ having $n$ variables $m$ clauses and \textit{at most} $k$ literals per clause, outputs another formula $\Phi'$ having $n' = n$ variables $m' \leq m \sigma \log n$ clauses and \textit{at least} $\log \log n$ literals per clause. Then we will invoke $\mathscr{A}_1$ on $\Phi'$: to do that, we only have to make sure that the clauses of $\Phi'$ look randomly generated to $\mathscr{A}_1$, that is to say the last $\sigma n$ of them shall be indistinguishable by $\mathscr{A}_1$ from a random instance, whereas the first $m'-\sigma n$ shall only have the property that each variable is mentioned roughly the same number of times. Clearly, the same definitions introduced in the previous section also apply to the present section, unique couple of exceptions being, as already anticipated, the $\leq k$ assumption used here instead of the $=k$ assumption used there, and the fact that here we require $k$ to be constant whereas there we did not.

\theoremstyle{definition}
\begin{definition}[\textit{Random inflation} $\mathcal{R}_{c,z}$]
Given a clause $c \in \Phi$ and an integer $z>0$, $\mathcal{R}_{c,z}$ is a randomly generated set of clauses $c_1, \cdots, c_{2^z}$, each one having exactly $|c|+z$ literals. Being $V_c$ the set of variables mentioned in $c$, the generation of $\mathcal{R}_{c,z}$ consists in randomly picking $z$ variables from $V \setminus V_c$ and in building the $2^z$ inflated clauses by adding $z$ literals to $c$, one such clause for each possible combination of signs. 
\end{definition}

\theoremstyle{definition}
\begin{definition}[\textit{Random inflation} $\mathcal{R}_{\Phi,z}$]
$\mathcal{R}_{\Phi,z} = \bigcup_{c \in \Phi} \mathcal{R}_{c,z}$. 
\end{definition}

\noindent Thus given a generic $\Phi$, we randomly inflate it by randomly inflating each of its clauses. Let us state the obvious: the $z$ variables used to inflate $\Phi$ are re-picked again and again for each clause (and of course thrown back in the $V$ basket after being used), in order for each set $\mathcal{R}_{c,z}$ to appear random to each other. 

\begin{lemma}\label{samesat}
$\Phi$ and $\mathcal{R}_{\Phi,z}$ have the same set of satisfying assignments.
\end{lemma}
\begin{proof}
Each set of clauses $\mathcal{R}_{c,z}$ implies $c$, by applying $2^z-1$ resolution steps to the $z$ literals. The clause $c$ we get back subsumes every clause in $\mathcal{R}_{c,z}$, which thus all disappear by leaving $c$ only. Repeating this process for every $c$ let us transform $\mathcal{R}_{\Phi,z}$ back into $\Phi$ in $m(2^z-1)$ resolution steps.  
\end{proof}
\noindent Clearly such proof relies on the fact that the $z$ variables which were used to inflate clauses were all already mentioned in $\Phi$: none of them was a fresh new variable (in which case the above Lemma \ref{samesat} would have been false). So far we have constructed, from our input formula $\Phi$ having $n$ variables $m$ clauses and at most $k$ literals per clause, another random looking formula having $n$ variables $2^z m$ clauses and at least $z$ literals per clause. We are now ready to present our randomized reduction:
\renewcommand{\thealgorithm}{}
\begin{algorithm}[H]
\caption{$\mathscr{R}$ Reduces $\Phi$ to $\Phi^{'}$}
\label{alg:R}
\begin{algorithmic}[1]
\Procedure{Inflate}{$\Phi$}
    \State Initialize $\Phi'_{\uparrow} = \Phi'_{\downarrow} = \varnothing$
    
    \For {each $c \in \Phi$}
            \State Let $\mathcal{R}_c = \mathcal{R}_{c, \log\log n}$
            \State Let $c_\downarrow \in \mathcal{R}_c$ be randomly picked
            \State $\Phi'_{\uparrow} \gets \Phi'_{\uparrow} \cup \mathcal{R}_c \setminus \{c_\downarrow\}$
            \State $\Phi'_{\downarrow} \gets \Phi'_{\downarrow} \cup \{ c_\downarrow \}$
    \EndFor
	
	\State Return $\Phi'_{\uparrow} \cup \Phi'_{\downarrow}$
\EndProcedure
\Statex
\Procedure{Inflate}{$\Phi$, $\sigma$}
\If {$m \geq \sigma n$}
    \State Return \Call{Inflate}{$\Phi$}
\Else
    \State Initialize $\Phi'_{\uparrow} = \Phi'_{\downarrow} = \varnothing$
    \For{each $i = 1, \cdots, \frac{\sigma n}{m}$}
	    \State Let $\Phi'_{i} = \Call{Inflate}{\Phi}$
	    \State $\Phi'_{\uparrow} \gets \Phi'_{\uparrow} \cup \Phi'_{i,\uparrow}$
	    \State $\Phi'_{\downarrow} \gets \Phi'_{\downarrow} \cup \Phi'_{i,\downarrow}$
	\EndFor
	\State Return $\Phi'_{\uparrow} \cup \Phi'_{\downarrow}$
\EndIf    
\EndProcedure
\end{algorithmic}
\end{algorithm}
\noindent At line 9 above, we mean that $\Phi'_\downarrow$ are the \textit{last} $m$ clauses of the returned $\Phi'$ (similarly for $\Phi'_{i, \downarrow}$ at line 19). As promised in the beginning of this section, such reduction $\mathscr{R}$ allows us to obtain, from any input formula $\Phi$ having at most a constant number $k$ of literals per clause, another formula $\Phi'$ having as many variables as $\Phi$, a number of clauses at most $\sigma \log n$ times higher, at least $\log \log n$ literals per clause, and exactly the same set of satisfying assignments as $\Phi$. Moreover, the sub-formula composed by the last $\sigma n$ clauses of $\Phi'$, denoted as $\Phi'_{\downarrow}$, is a full-fledged random formula: from the point of view of $\mathscr{A}_1$, it will behave for all intents no differently than a random CNF instance having exactly $\log\log n$ literals per clause. The randomized algorithm $\mathscr{A}_2$ which makes us able to solve \#$k$-SAT in sub-exponential time is therefore the following:

\renewcommand{\thealgorithm}{}
\begin{algorithm}
\caption{$\mathscr{A}_2$ Computes the exact number of satisfying assignments of $\Phi$}
\label{alg:A2}
\begin{algorithmic}[1]
\Procedure{Count}{$\Phi$, $\sigma$}
    \State Let $\Phi' = \Call{Inflate}{\Phi, \sigma}$
    \State Return \Call{CountRandom}{$\Phi'$, $\sigma$} 
\EndProcedure
\end{algorithmic}
\end{algorithm}
\noindent We are now ready to complete the paper by finally proving our main result, as follows:
\begin{theorem}
$\mathscr{A}_2$ runs in $2^{o(n)}$ time.
\end{theorem}
\begin{proof}
We only have to show that the proof of Theorem \ref{the:A1} goes through as well with $\Phi'$. In order to do so, we can consider the $k$ in Theorem \ref{the:A1} to be equal to $\log \log n$. So let us read that proof again and check if every step of it stands valid with $\Phi'$. First step is to verify whether the upper bound \ref{eq:cardMax} on the number of clauses that any critical $\Psi$ might possibly have still holds: it is evident that it does, for its underlying hypothesis (that at least $\frac{k}{2}=\frac{\log \log n}{2}$ new variables are added to $\Psi_\uparrow$ for each new clause inserted in it) is still valid, due to the fact that, thanks to the \textit{non-constant} $z$ random inflation, any variable occurs in $\Phi'$ roughly the same number of times as any other variable (whereas it might not be the case in $\Phi$). We also observe that the critical saturation is $s_\Psi = 4 n^{-\frac{1}{\log\log n}}$, which means $s_\Psi n \in o(n)$. By plugging $k=\log\log n$ into \ref{eq:cardMax}, we can write the following:
\begin{equation}
|\Psi_\uparrow| \leq \frac{8 n^\frac{\log\log n-1}{\log\log n}}{\log\log n}
\end{equation}
By our usual notation, being $\Delta$ (respectively $\Delta'$) the clause density of $\Phi$ (respectively $\Phi'$), we can write:
\begin{equation}\label{eq:Delta'}
\Delta' \leq \Delta \sigma 2^{\log\log n} = \Delta \sigma \log n
\end{equation}
By plugging \ref{eq:Delta'} into \ref{eq:outer}, and displaying the leading term only, we can write:
\begin{equation}\label{eq:outer_A2}
\log { (\Delta' - \sigma) n \choose \frac{8 n^\frac{\log\log n-1}{\log\log n}}{\log\log n} } \approx \frac{8 n^\frac{\log\log n-1}{\log\log n}}{\log\log n} \log \frac{\log\log n(\overbrace{\Delta\sigma\log n}^{\geq \Delta'} - \sigma)n^\frac{1}{\log\log n}}{8}
\end{equation}
\noindent We observe how $n^\frac{\log \log n - 1}{\log\log n} \in o(n)$ and how the increased $\Delta'$ does not behave substantially different than the original $\Delta$ in terms of its impact on the overall expression in \ref{eq:outer_A2}, which as a whole remains sub-exponential as well in any case. This means that:
\\
\begin{mdframed}[backgroundcolor=gray!7] 
\vspace*{\fill} 
\begin{quote} 
\centering 
\textit{The outer iteration of $\mathscr{A}_1$ cycles $2^{o(n)}$ many times also when fed with $\Phi'$.}
\end{quote}
\vspace*{\fill}
\end{mdframed}
\vspace{10pt}

\noindent We repeat here the same observation made in the previous section: we are overestimating the number of loops of the outer cycle, because we are pretending that every clause can coexist in $\Psi_\uparrow$ with every other clause (see how this is clearly false for every pair of clauses picked from the same $\mathcal{R}_c$). We shall now focus on the last $\sigma n$ clauses of $\Phi'$: thanks to our randomized reduction $\mathscr{R}$, such last clauses have been built and arranged in such a way to be indistinguishable by $\mathscr{A}_1$ from a random \#$\log\log n$-SAT instance. What do we mean by "indistinguishable by $\mathscr{A}_1$"? We mean that the crucial property exploited by $\mathscr{A}_0$ (and re-used by $\mathscr{A}_1$) holds: asymptotically, every monotone sub-formula assembled using the last $\sigma n$ clauses has cardinality at most $\frac{n}{\log \log n}$, as long as $\sigma \in 2^{o(\log\log n)}$ (which is obviously true as $\sigma$ is a constant). This means we can re-apply that very same argument and conclude that the number of maximal monotone sub-formulae to be perlustrated among the last $\sigma n$ clauses by the inner iteration of $\mathscr{A}_1$ behaves as $2^{\varepsilon n}$, where $\varepsilon = \frac{\log \log\log n}{\log \log n} + \frac{\log \sigma}{\log \log n} + \frac{1}{\log \log n} - \frac{1}{\sigma \log^2 \log n}$, with the first term clearly being the leading term: 
\\
\begin{mdframed}[backgroundcolor=gray!7] 
\vspace*{\fill} 
\begin{quote} 
\centering 
\textit{When fed with $\Phi'$, the inner iteration of $\mathscr{A}_1$ cycles $2^{O\left(\frac{\log\log\log n}{\log\log n}n\right)}$ many times.}
\end{quote}
\vspace*{\fill}
\end{mdframed}
\vspace{10pt}

\noindent The proof ends the same way it ended in the previous section: the running time of $\mathscr{A}_2$, which is given by the running time of $\mathscr{A}_1$ when fed with $\Phi'$, is the product of the outer and inner iterations, both of them executing a sub-exponential amount of cycles. The exponent of the overall running time is the sum of the two exponents of those iterations, now both of them being $o(n)$:
\[\Scale[1.25]
{
\Scale[2.10]{2}^{ \overbrace{ \underbrace{\frac{8 \overbrace{n^\frac{\log\log n-1}{\log\log n}}^{o(n)}}{\log\log n} \log \frac{\log\log n(\Delta\sigma\log n - \sigma)n^\frac{1}{\log\log n}}{8}}_{\text{Negligible term, outer iteration}} }^{o(n)\text{ for any }\Delta} +\  \Scale[1.35]{\underbrace{\overbrace{ O\left(\frac{\log\log\log n}{\log\log n}n\right)}^{o(n)}}_{\substack{\text{Leading term,} \\ \text{inner iteration}}\\}}}
}
\]
\noindent As $n \to \infty$, by definition the above can be rewritten as:
\[\Scale[2.1]
{2^{O\left(\frac{\log\log\log n}{\log\log n}n\right)}\in 2^{o(n)}}
\]

\noindent thereby concluding the proof. We have a randomized counting algorithm which returns the exact number of satisfying assignments with probability $1$, running in $2^{o(n)}$ time on generic $k$-SAT instances, for any clause density and any constant $k$.
\end{proof}

\section*{References}

\printbibliography[heading=none]

\end{document}